\documentclass{article} 
  
\usepackage{float}
\usepackage{amsmath, amstext, amsthm, amssymb, amsfonts, mathtools}
\usepackage[]{algorithm2e, algorithmicx, algpseudocode} 
\usepackage{color}
\usepackage{enumitem} 
\usepackage{wrapfig}
\usepackage[margin=1in]{geometry}
\usepackage[font=small,labelfont=bf]{caption}
\usepackage{comment}
\usepackage{placeins}
\usepackage{url}
\usepackage{graphicx}

\usepackage{sidecap}
\sidecaptionvpos{figure}{c}

\newcommand{\defn}[1]{\textbf{#1}}

\newcommand{\st}[3]{\ensuremath{s(#1,#2,#3)}}
\renewcommand{\paragraph}[1]{\subsubsection*{#1}}

\newtheorem{theorem}            {Theorem}
\newtheorem{lemma}[theorem]     {Lemma}

\newtheorem{definition}[theorem]{Definition}

\newcommand{\filterint}{(\x\cap\q)\setminus v}
\newcommand{\dis}{\mathcal{D}}
\DeclareMathOperator{\Prob}{Pr}
\DeclareMathOperator{\E}{\mathbf{E}}
\newcommand{\x}{\mathbf{x}}
\newcommand{\y}{\mathbf{y}}
\newcommand{\z}{\mathbf{z}}
\newcommand{\q}{\mathbf{q}}
\newcommand{\n}{\mathbf{n}}
\newcommand{\concat}{\circ}
\renewcommand{\hat}{\widehat}
\newcommand{\stx}{\st{\x}{j}{i}}

\newcommand{\sam}[1]{{\color{blue}{\bf \sf \scriptsize Sam:} \sf \scriptsize #1}}

\DeclarePairedDelimiter{\ab}{\lvert}{\rvert}

\title{Set Similarity Search for Skewed Data\footnotetext{
The research leading to these results has received funding from the {European Research Council under the European Union's 7th Framework Programme (FP7/2007-2013)}{} / ERC grant agreement no.~{614331}.
        BARC, Basic Algorithms Research Copenhagen, is supported by the VILLUM Foundation grant 16582.}
      }

      \date{}

\author{Samuel McCauley
\footnote{
    {BARC and IT U.~Copenhagen},
    {Copenhagen},
    {Denmark}.
  \protect\url{{samc,pagh}@itu.dk}
  }
\and
Jesper W. Mikkelsen
\footnote{
    {IT U.~Copenhagen},
    {Copenhagen},
    {Denmark}.
\protect\url{jesperwm@gmail.com}
  }
\and
Rasmus Pagh\footnotemark[2]}

\begin{document}
\maketitle


\begin{abstract}
\emph{Set similarity join}, as well as the corresponding indexing problem \emph{set similarity search}, are fundamental primitives for managing noisy or uncertain data.
For example, these primitives can be used in data cleaning to identify different representations of the same object. 
In many cases one can represent an object as a sparse 0-1 vector, or equivalently as the set of nonzero entries in such a vector.
A set similarity join can then be used to identify those pairs that have an exceptionally large dot product (or intersection, when viewed as sets).
We choose to focus on identifying vectors with large \emph{Pearson correlation}, but results extend to other similarity measures.
In particular, we consider the indexing problem of identifying correlated vectors in a set $S$ of vectors sampled from $\{0,1\}^d$.
Given a query vector $\y$ and a parameter $\alpha\in(0,1)$, we need to search for an $\alpha$-correlated vector $\x$ in a data structure representing the vectors of $S$.
This kind of similarity search has been intensely studied in worst-case (non-random data) settings.

Existing theoretically well-founded methods for set similarity search are often inferior to heuristics that take advantage of \emph{skew} in the data distribution, i.e., widely differing frequencies of 1s across the $d$ dimensions.
The main contribution of this paper is to analyze the set similarity problem under a random data model that reflects the kind of skewed data distributions seen in practice, allowing theoretical results much stronger than what is possible in worst-case settings.
Our indexing data structure is a recursive, data-dependent partitioning of vectors inspired by recent advances in set similarity search.
Previous data-dependent methods 
do not seem to allow us to exploit skew in item frequencies, so we believe that our work sheds further light on the power of data dependence.
\end{abstract}


\section{Introduction}

Data management is increasingly moving from a world of well-ordered, curated data sets to settings where data may be noisy, incomplete, or uncertain.
This requires primitives that are able to work with notions of ``approximate match'', as opposed to the exact matches used in standard hash indexes and in equi-joins.
Such functionality is particularly challenging to scale when data is high-dimensional, informally because of the \emph{curse of dimensionality} which makes it hard to organize data in such a way that approximate matches can be efficiently identified.
In this paper we consider the fundamental primitive \emph{set similarity join} and more specifically the indexing problem \emph{set similarity search}.
A set similarity join is used to identify pairs of sets that are similar in the sense that they have an ``exceptionally large intersection''.
Many notions of ``exceptionally large intersection'' exist, but from a theoretical point of view they are essentially equivalent~\cite{christiani2017set}.
We will consider the encoding of sets as sparse 0-1 vectors, and use Pearson correlation as the measure of similarity; we give more details below.

The information retrieval and database communities have extensively worked on designing scalable algorithms for similarity join, see e.g.~the recent book by Augsten and B{\"o}hlen~\cite{augsten2013similarity}.
The state-of-the-art for practical set similarity join algorithms is reflected in the recent mini-survey and comprehensive empirical evaluation of Mann et al.~\cite{Mann2016}.
The best methods in practice are ones that exploit the significant \emph{skew} in frequencies of set elements that exists in many data sets.
When there is insufficient skew these methods become inefficient, and in the worst case they degenerate to a trivial brute-force algorithm.
In contrast, strong theoretical results are known when \emph{randomization} and \emph{approximation} of distances is allowed, e.g.~\cite{DBLP:conf/pods/AhlePR016,christiani2017set,DBLP:conf/pods/HuTY17,Kapralov15,DBLP:conf/kdd/Zhang017}.
Even though existing randomized algorithms for set similarity search are superior to commonly used heuristics for difficult data distributions with small skew, it is clear that the heuristics will work much better (in theory and in practice) when the skew is large enough.

In this paper we target this disconnect between theory and practice, presenting a new data structure that, in a certain way, can \emph{interpolate} between the best existing methods for small skew and heuristics that work well with large skew.
Our message is that modeling skew can lead to algorithms that takes advantage of structure in data in a way that is theoretically justified.
This complements recent advances in the theory of \emph{data dependent} methods for high-dimensional search, where clustering structure in data (rather than skew) is exploited to achieve faster algorithms, even in the worst case~\cite{andoni2014beyond,andoni2015optimal}.

\paragraph{Motivating example.}
Suppose we wish to search $n$ $d$-dimensional boolean vectors $\x^1,\ldots, \x^n$ chosen from the ``harmonic'' distribution where the $k$th bit is set with probability $\Pr[\x^j_k = 1] = 1/k$, independently for each $\x^j$ and each $k\in\{1,\dots,d\}$.  The boolean vectors represent subsets of $\{1,\ldots, d\}$.  For consistency with the rest of the paper, we refer to these elements as vectors, but use some set notation for simplicity; i.e. $|\x|$ is used to represent the Hamming weight of $\x$.

Assuming (for now) $\log d\gg \log n$, all vectors have Hamming weight close to the expectation $|\x|\approx \sum_{k=1}^d 1/k \approx \ln d$ with high probability by Chernoff bounds.
We wish to search for a vector $\x^{j^*}$ that is correlated with a query vector $\q$ such that $|\x^{j^*}\cap \q|\geq i_1 |\q|$, for some parameter $i_1\in (0,1)$.

Define $i_2 |\q| = \E[|\x^j\cap \q|] = \sum_{k\in \q} 1/k$ to be the expected intersection size between $\q$ and $\x^j$ for $j\ne j^*$, $i_2 < i_1$.
Assuming $i_2 |\q| \gg \log n$ we will have $|\x^{j}\cap \q|\approx i_2 |\q|$ for all $j\ne j^*$ with high probability.
In this setting it is known how to perform the search in expected time roughly $n^{\rho}$, where $\rho = \log(i_1)/\log(i_2)$, and this bound is tight for LSH-like techniques~\cite{christiani2017set}.

However, we can do better for skewed distributions by splitting the search problem into two parts:
split $\q$ into two equal-sized vectors 
$\q^\text{frequent}$ and $\q^\text{rare}$, defined as the first (resp. last) half $d/2$ bits of $\q$.
Note that for every choice of parameter $\ell$, we either have $|x^{j^*}\cap \q^\text{frequent}| \geq \ell |\q|$ or
$|x^{j^*}\cap \q^\text{rare}| \geq (i_1-\ell) |\q|$, so the original search problem can be solved by performing searches for a set with a large overlap either with $\q^\text{frequent}$ or $\q^\text{rare}$.
Let 
$$i_\text{frequent} = \E[|\x^{j}\cap \q^\text{frequent}|]/|\q|,\text{ and}$$
$$i_\text{rare} = \E[|x^{j}\cap \q^\text{rare}|]/|\q|,$$
such that $i_2 = i_\text{frequent} + i_\text{rare}$.
Now the combined cost of the two searches becomes approximately 
$n^{\rho_\text{frequent}}+n^{\rho_\text{rare}}$, where
$$\rho_\text{frequent} = \log(\ell)/\log(i_\text{frequent})\text{, and}$$
$$\rho_\text{rare} = \log(i_1-\ell)/\log(i_\text{rare}) \enspace .$$
Choosing $k$ to balance the two terms we get a faster query time whenever $i_\text{frequent} \gg i_\text{rare}$, i.e., when the distribution of elements in $\q$ has large skew.

The example shows that skew can be exploited, but it remains unclear how to do this in a principled way. This question was the starting point for this paper.

\paragraph{Probabilistic viewpoint.}
Establishing correlation between random variables is a fundamental and well-studied problem in statistics, but computational aspects of this problem are far from settled.
In a breakthrough paper~\cite{Valiant15}, Greg Valiant addressed the so-called \emph{light bulb problem} originally posed in~\cite{valiant1988functionality}: 

\medskip

\emph{Given a set of $n$ vectors $S\subseteq \{0,1\}^d$ chosen uniformly and independently at random, with the exception of one pair of distinct vectors $\x,\y\in S$ that have correlation $\alpha > 0$, identify the vectors $\x$ and $\y$.}

\medskip

If $d$ is sufficiently high (e.g.~$d\gg \log(n)/\alpha$) the correlated vectors $\x$ and $\y$ are, with high probability, the only pair with an inner product of around $(1+\alpha)d/4$, while all other pairs have inner product around $d/4$.
From an algorithmic perspective the problem then becomes that of finding the pair of vectors with a significantly higher inner product.
The problem of searching for correlated vectors has been intensely studied in recent years, in theory~\cite{karppaSODA16,DBLP:conf/esa/KarppaKKC16,alman2015probabilistic,DBLP:conf/focs/AlmanCW16,DBLP:conf/pods/AhlePR016,abboud2017distributed,DBLP:conf/eurocrypt/0001O15} and in practice~\cite{ShrivastavaL14,shrivastava2015asymmetric,DBLP:conf/uai/Shrivastava015,RamG12,DBLP:conf/ijcnn/KeivaniSR17,DBLP:journals/algorithmica/DasguptaS15,DBLP:conf/colt/DasguptaS13,teflioudi2015lemp}.
Practical solutions often address the \emph{search} version, known as \emph{maximum inner product search}, where the vector $\y$ is given as a query (often denoted $\q$) and the task is to search for the correlated vector $\x$ in a data structure enabling fast search among the vectors in $S$.

The light bulb problem is perhaps the cleanest and most fundamental correlation search problem.
However, vectors of real-life data sets are usually not well described by a uniform distribution over $\{0,1\}^d$.
Instead, such vectors are often sparse (assuming without loss of generality that 0 has probability at least $1/2$ in each coordinate), and the fraction of vectors having the value 1 in the $i$th coordinate varies greatly with $i$, often following e.g.\ a Zipfian distribution (see Section~\ref{sec:experiments}).
This kind of \emph{skew} is exploited by practical solutions to correlation search~\cite{Bayardo_WWW07,wang2012can,Mann2016}, since high correlation between vectors $\x$ and $\y$ will often be ``witnessed'' by $\x_i = \y_i = 1$, where the set $\{ \z\in S \; | \; \z_i = 1\}$ is small.
On the other hand, such methods do not perform well when the skew is small.
In this paper we explore the computational problem of identifying correlations in random data with skew, focusing on the search version of the problem.
Generalizing and modifying recent worst-case efficient data structures, we are able to get a smooth trade-off between ``hard'' queries and data sets with no skew, and ``easy'' queries and data sets of the kind often encountered in practice.

To model skewed data we adopt the model of Kirsch et al.~\cite{kirsch2012efficient} that was previously used to give statistical guarantees on data mining algorithms.
We are not aiming at statistical guarantees, but instead use this model as an interesting ``middle ground'' between uniformly random and worst-case settings when analyzing algorithms dealing with high-dimensional data.
Conceptually this model:
\begin{itemize}
\item is expressive enough to model real-world data (such as the feature vectors that are ubiquitous in machine learning) much better than random data, 
\item avoids the pessimism of worst-case analysis, yet 
\item is simple enough to be tractable to analyze.
\end{itemize}

\subsection{Our Results}

We assume that data vectors are sampled from a distribution $\dis$ (see Section~\ref{sec:model} for details).
Let $S$ be the set of $n$ data vectors sampled independently from $\dis$.
We parameterize how close a query is to its nearest neighbor using a parameter $\alpha$.
For $\alpha > 0$ a query $\q$ that is $\alpha$-correlated with some $\x\sim\dis$ can be defined as follows:
Let $\n\sim\dis$ be a ``noise vector'', and independently let
\[
\q_i = \begin{cases}
    x_i & \text{ with probability } \alpha\\
    n_i & \text{ with probability } 1 - \alpha \enspace .
\end{cases}
\]
For a data vector $\x\sim\dis$ we define $p_i = \Pr[\x_i=1]$.
We assume $p_i < 1/2$ for all $i$ and that each bit of $\x$ is sampled independently.  

Our results involve two additional parameters.  
First, we follow previous literature (e.g. ~\cite{christiani2017set, christiani2017framework,DBLP:conf/soda/AndoniLRW17}) in parameterizing our running time by a constant $\rho$.  Generally, $\rho$ would only depend on the similarity of the planted close points ($\alpha$ in this case); for our problem $\rho$ is a function of $\alpha$ and $\dis$ (and the query in the adversarial case).
Secondly, we assume that there is a large constant $C$ satisfying $\sum_{i\in[d]} p_i = C\log n$.  We require $C$ to be large both to ensure correctness\footnote{Specifically, to ensure that our data structure is correct with high probability.} and to achieve our target running time.

\begin{theorem}\label{thm:alpha-correlated}
    Consider a dataset $S$ of $n$ vectors sampled from $\dis$ and let $C$ satisfy $\sum_{i\in[d]} p_i \geq C\log n$.

    Assume that $\q$ is $\alpha$-correlated with $\x$ for some $\alpha > 0$.

	Our data structure returns $\x$ on query $\q$ with high probability.   
    Let $\rho$ satisfy 
    \[
        \sum_{i\in[d]} \frac{p_i^{1+\rho}}{p_i (1-\alpha) + \alpha)} = \sum_{i\in[d]} p_i.
    \]
Then for every $\epsilon > 0$ there exists a sufficiently large $C$ such that each query has expected cost $O(dn^{\rho + \epsilon})$, and the data structure requires expected $O(n^{1 + \rho + \epsilon} + dn)$ space.  
\end{theorem}
\paragraph{Discussion.}
In the balanced case where all probabilities $p_i$ are identical, we recover the time bounds of the recently proposed {\sc ChosenPath} algorithm~\cite{christiani2017set}, which are known to be optimal in this setting.
In the very unbalanced case where some $p_i$ are $\Omega(1)$, some $p_i$ are $O(1/n)$, and the expected number of items of both kinds are comparable, we match the well-known \emph{prefix filter} algorithm~\cite{Bayardo_WWW07}, which beats {\sc ChosenPath} in this setting.
For skew between these extremes we get strict improvements over existing methods.
(See Section~\ref{sec:examples} and Figure~\ref{figure} for further discussion.)
%
\begin{wrapfigure}{r}{.5\textwidth}
\centering
  \vspace{.18in}
\includegraphics[width=0.40\textwidth]{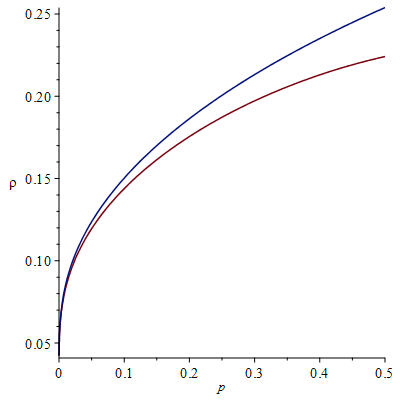}
      \caption{The red line gives the $\rho$ value of our data structure when the distribution is such that half the bits are set to $1$ with probability $p$ and the other half is set to $1$ with probability $p/8$, and the sought-for correlation is $\alpha=2/3$. The blue line gives the $\rho$-value achieved by Chosen Path. Prefix filtering has a $\rho$-value of $1$ in this case and therefore not included in the figure. We see that even though Chosen Path achieves the optimal $\rho$-value for solving the $(b_1,b_2)$-approximate similarity problem when considering worst-case inputs, we can achieve a better $\rho$-value when the input distribution is skewed.} 
      \label{figure}
      \vspace{-.5in}
  \end{wrapfigure}~
\vspace*{-\baselineskip}
\paragraph{Techniques.}
Our data structure is a natural, recursive variant of {\sc ChosenPath} that is able to exploit skew by varying the recursion depth over the branches of the recursion tree and by aggressively favoring choices based on the given distribution that are more likely distinguish close and far elements. 
We stress that {\sc ChosenPath} is not able to exploit skew, and in fact has the same tight running time guarantee independent of the data distribution.
Because we cut the depth of the tree earlier based on the skew of the distribution, we must tighten the previous analysis to handle sampling without replacement, while also parameterizing our performance based on skew.  This leads to significant challenges.


\paragraph{Adversarial queries.}
It may not always be reasonable to assume that a query is random and $\alpha$-correlated, as in Theorem~\ref{thm:alpha-correlated}.
For this reason we analyze the setting where the query may be adversarially chosen.
In Section~\ref{sec:datastructure} we give a data structure that adapts to the difficulty of the query, matching existing worst-case bounds for ``hard'' queries, while being much faster for ``easy'' queries.
Our result uses a similarity function $B(\x,\q)$ (defined later) to parameterize how close the query is to the nearest element of the data set.

\begin{theorem}\label{thm:adversarial}
    Consider a dataset $S$ of $n$ vectors sampled from $\dis$ and let $C$ satisfy $\sum_{i\in[d]} p_i \geq C\log n$.
    Let $\rho_u$ satisfy
    \[
        \sum_{i\in [d]} p_i^{1 + \rho_u} = b_1\sum_{i\in [d]} p_i.
    \]
    For any query $\q$ of size $\ab{\q} \geq C\log n$ satisfying $B(\q,\x)\geq b_1$ for some $\x\in S$, let $\rho(q)$ satisfy 
\[
    \sum_{i\in \q} p_i^{\rho(q)} = b_1\sum_{i\in \q} p_i.
\]
    Then for any $\varepsilon > 0$ there exists a sufficiently large $C$ such that 
    our data structure requires $O(n^{1+\rho_u + \varepsilon} + dn)$ expected space, can be built in $O(dn^{1 + \rho_u + \varepsilon})$ time, and can perform a search in time $O(dn^{\rho(q) + \varepsilon})$, returning $\x$ with probability\footnote{This probability can be increased using independent repetitions.} at least $1/2$.
\end{theorem}

\paragraph{Similarity joins.}
Our results immediately apply to the problem of database similarity joins~\cite{PaghPSS15,Arasu_VLDB06,li2011pass,Xiao_WWW08,augsten2013similarity,Mann2016,DBLP:conf/pods/AhlePR016,DBLP:conf/pods/HuTY17}.  

Many similarity join algorithms work using (essentially) repeated similarity search queries; see e.g.~\cite{Mann2016,PaghPSS15,jiang2013efficient}.  This method is equally effective here.  Assume that we want to find all similar pairs between sets $R$ and $S$, and that the actual join size (i.e. the number of close pairs) is much smaller than $R$ or $S$.\footnote{If the join is large we can extend this idea via parameterizing by the join size; see e.g. the time bounds in~\cite{PaghPSS15,DBLP:conf/pods/HuTY17}.}  Then Theorem~\ref{thm:adversarial} implies that we can preprocess $S$ in time $O(d|S|^{1 + \rho})$ and find all pairs in time $O(d|R||S|^{\rho})$.  The same idea extends to Theorem~\ref{thm:alpha-correlated} as well.

\subsection{Related work}

Assume in the following that $d$ is large enough that the empirical correlation closely matches the true correlation.

\paragraph{Correlation search on the unit sphere.}
After shifting vectors to have zero mean and normalizing them (which does not affect correlation), searching for a vector $\x\in S$ that is $\alpha$-correlated with a query vector $\y$ boils down to a search problem on the unit sphere $S^{d-1}$:
Given $\y\in S^{d-1}$ find $\x\in S$ such that $\x\cdot \y \geq \alpha$.
In turn, this is equivalent to near neighbor search under Euclidean distance on the unit sphere, which has been studied extensively.
For ``balanced'' data structures with space (and construction time) $\tilde{O}(n^{1+\rho})$ and query time $\tilde{O}(n^{\rho})$ the state-of-the-art is $\rho=\tfrac{1-\alpha}{1+\alpha}$~\cite{andoni2015practical}, using a non-trivial generalization of Charikar's famous LSH for angular distances~\cite{Cha02}.
Generalizations to various time-space trade-offs are also known~\cite{christiani2017framework,DBLP:conf/soda/AndoniLRW17}.

\paragraph{Correlation search on sparse vectors.}
For sparse binary vectors, the framework of \emph{set similarity search} captures search for different notions of similarity/correlation~\cite{choi2010survey}.
For vectors of fixed Hamming weight there is a 1-1 correspondence between Pearson correlation and standard set similarity measures such as Jaccard similarity and Braun-Blanquet similarity.
It is known that methods such as MinHash~\cite{bro97a,bro97b} that specifically target sparse vectors yield better search algorithms than more general methods when the fraction $\beta$ of non-zero entries is bounded by a sufficiently small constant.
Recently it was shown that MinHash can be improved in this setting, yielding balanced similarity search data structure with $\rho=\log(\beta+\alpha(1-\beta))/\log \beta$~\cite{christiani2017set}.

\paragraph{Batched search.}
The first subquadratic algorithm for the light bulb problem was given by Paturi et al.~\cite{DBLP:conf/colt/PaturiRR89}, yielding time $n^{2-\Theta(\alpha)}$.
Valiant~\cite{Valiant15} showed that it is possible to remove $\alpha$ from the exponent,  obtaining substantially subquadratic running time even for small values of $\alpha$.
Karppa et al.\ strenthened Valiant's result to time $O(n^{1.6} \text{poly}(1/\alpha))$~\cite{karppaSODA16}, for $\log n \ll d = n^{o(1)}$ (where the constant 1.6 reflects current matrix multiplication algorithms).
A somewhat slower, but deterministic, subquadratic algorithm was subsequently presented~\cite{DBLP:conf/esa/KarppaKKC16}.
With the exception of~\cite{DBLP:conf/colt/PaturiRR89} these methods rely on fast rectangular matrix multiplication and seem inherently only applicable to batched search problems.

\paragraph{Worst-Case search.}  Much of previous theoretical work on set similarity search has focused on the worst case, where the point set and queries are adversarial rather than being drawn from a known distribution.  The Chosen Path algorithm of Christiani and Pagh 
focuses on Braun-Blanquet similarity, which is the metric we use here.\footnote{We show that Pearson correlated vectors are very likely to have high Braun-Blanquet similarity in Lemma~\ref{lem:correlatedclose}.}  In particular, if the dataset contains a point $p$ where $q$ and $p$ have Braun-Blanquet similarity at least $b_1$, the data structure returns a point $p'$ such that $q$ and $p$ have Braun-Blanquet similarity at least $b_2$.   Their data structure 
achieves space $\tilde{O}(n^{1 + \rho})$ and 
query time $\tilde{O}(n^{\rho})$ for $\rho = (\log b_1)/\log b_2$.  This result is a strict improvement over the classic MinHash algorithm~\cite{bro97a,bro97b}.  

\paragraph{Lower bounds.}
Cell probe data structure lower bounds achievable by current techniques are polylogarithmic in data size (assuming polylogarithmic word size).
Tight bounds are only known for small query time, see e.g.~\cite{panigrahy2010lower} --- for subquadratic space this leaves a large gap to the polynomial upper bounds that can be achieved by similarity search techniques.
However, good \emph{conditional} lower bounds based on the strong exponential time hypothesis (SETH)~\cite{Impagliazzo_JCSS01} are known.
Ahle et al.~\cite{DBLP:conf/pods/AhlePR016} showed hardness of approximate maximum inner product search under SETH, but only for inner products very close to zero.
Recently, Abboud, Rubinstein, and Williams~\cite{abboud2017distributed} significantly improved this result, and in particular showed, again assuming SETH, that maximum inner product requires near-linear time even for vectors in $\{0,1\}^d$ and allowing a large $2^{(\log d)^{1-o(1)}}$ approximation factor.
This means that there is little hope of obtaining strong algorithms in the worst case, even for batched search problems.
Of course, as the upper bounds for the light bulb problem show, the average case is easier.

\paragraph{Heuristics.}
Above we have discussed related work with a \emph{theoretical} emphasis.
Many heuristics exist, especially for sparse vectors.
Some of the most widely used ones are based on exact, deterministic filtering techniques, such as prefix filtering~\cite{Bayardo_WWW07} which evaluates the similarity of every pair of vectors that share a $1$ in a position that has a small fraction of $1$s.
This is effective when vectors are sparse and \emph{skewed} in the sense that there are many positions with a small fraction of $1$s.
We refer to Mann et al.~\cite{Mann2016} for an overview of exact, heuristic techniques.

\section{Model}\label{sec:model}

Our model closely follows the one in~\cite{kirsch2012efficient} --- we describe it here for completeness.
The elements in our set are $d$-dimensional boolean vectors $\x\in\{0,1\}^d$.  These vectors represent a subset of items from a universe $U = \{1,\ldots,d\}$.  We generally consider our elements to be boolean vectors; however, we express some of our operations using set notation for convenience (for example, $\x \cap\q$ represents the 1 bits in common between $\x$ and $\q$).

We are given a distribution $\dis=\dis[p_1,\ldots , p_d]$ over $\{0,1\}^d$ defined 
as follows: if $\x$ is a vector drawn from $\dis$, then $\Prob[\x_i = 1] = p_i$ independently for each $i\in [d]$. We denote by $\dis^n=\dis^n[p_1,\ldots , p_n]$ the distribution obtained by sampling $n$ vectors independently from $\dis$. The probabilities $p_1,\ldots , p_d$ are assumed to be known to the algorithm. We assume that all item-level probabilities are at most $1/2$. The particular value $1/2$ is not important, and all of our results holds as long as there is some constant $M<1$ such that all item-level probabilities are bounded by $M$.

\begin{definition}[Correlation]
Let $\alpha\in[0,1]$.
    Fix $\x\in \{0,1\}^d$. We say that $\q$ is \emph{$\alpha$-correlated to $\x$ with respect to $\dis$}, written $\q\sim \dis_\alpha(\x)$, if $\q$ is a random vector drawn as follows: For each $i\in[d]$ independently, with probability $\alpha$ let $\q_i=\x_i$, and with probability $1-\alpha$ let $\q_i\sim \text{Bernoulli}(p_i)$.  
\end{definition}
If $\x\sim \dis$ and $\q\sim\dis_\alpha(\x)$, then the distribution of $\q$ is exactly $\dis$---in particular, $\Pr(\q_i=\nobreak 1)=\nobreak p_i$. Furthermore, for each
$i\in [d]$, the random variables $\q_i$ and $\x_i$ have Pearson correlation $\alpha$.

Our goal is to create an efficient data structure for vectors sampled from $\dis^n$ which takes 
advantage of possible skew in the item-level probabilities. Our performance bounds (query and 
preprocessing times) are expected, and 
depend both on $\dis^n$ and the data structure's random choices. Following~\cite{christiani2017set}, we will use Braun-Blanquet similarity,
\[
    B(\x,\q) = \frac{\ab{\x\cap\q}}{\max\{\ab{\x},\ab{\q}\}},
\]
as our measure for similarity between sets. This similarity measure is closely related to Jaccard similarity, see \cite{christiani2017set} for details. We now formally define the two versions of the problem that we are considering:

\paragraph{Adversarial query} Given a dataset $S\sim \dis^n[p_1,\ldots, p_d]$ of $n$ items sampled from $\dis$ and a similarity threshold 
$b_1$, preprocess $S$ into a data structure with the following capability: 
Given $\q\in \{0,1\}^d$, return $\x\in S$ such that $B(\x,\q)\geq b_1$ if such an $\x$ exists. The data 
structure must succeed with probability at least $1-o_n(1)$ over its own internal randomness. In particular, 
the probability that the data structure succeeds must be independent of the dataset $S\in\dis^n$.

\paragraph{Correlated query} Given a dataset $S\sim\dis^n[p_1,\ldots, p_d]$ of $n$ items sampled from $\dis$ and a correlation threshold 
$0<\alpha\leq 1$, preprocess $S$ into a data structure with the following capability: Let $\x\in S$ and let $\q\sim\dis_\alpha(\x)$ be $\alpha$-correlated to $\x$ with respect to $\dis$. Then, given the query $\q$, the data structure must return $\x$. The data structure must succeed in doing so with probability at least $1-o_n(1)$ over the choice of $\q$, the randomness of the dataset $S$, and its own internal randomness.\footnote{Clearly, some assumptions on the item-level probabilities of $\dis^n$ are needed for this to be possible. For all of our results, we will assume that $\sum_{i\in[d]}p_i$ is sufficiently large.}

As usual, the part of the success probability that depends only on the data structure's random choices 
may be boosted by a small number of repetitions. Thus, for adversarial queries, it suffices to build a data 
structure with success probability, say, $1/2$, and for correlated queries, it suffices that the part of the 
success probability depending only on the data structure's random choices is at least $1/2$. 
However, the part of the 
success probability depending on the query and the dataset cannot be boosted by 
repetitions. Instead, we need to design a data structure such that under some reasonable assumptions on the item-level probabilities, we achieve the desired success probability. 

We say that an event occurs \defn{with high probability} if it occurs with probability $O(1/n^c)$ for a tunable\footnote{We can make $c$ arbitrarily small by adjusting other parameters.  For example, in the above discussion, we can boost the probability of success from $1/2$ to $O(1/n^c)$ using $\Theta(c\log n)$ independent repetitions.} constant $c$.

We make use of the following weighted Chernoff bound found in e.g.~\cite[Ex. 4.14]{DBLP:books/daglib/0012859}.
\begin{lemma}
\label{wchernoff}
Let $X_1,\ldots , X_n$ be independent random variables. Let $a_1,\ldots , a_n\in [0,1]$, $p_1,\ldots , p_n\in [0,1]$, and assume that $\Pr [X_i=a_i]=p_i$, $\Pr [X_i=0]=1-p_i$ for $i\in [n]$. Furthermore, assume that $a_i\leq a$ for all $i\in [n]$, and let $S_n=\sum_{i=1}^nX_i$. Then,
\begin{align*}
    &\Pr\left[S_n \geq (1+\varepsilon)\E\left[S_n\right]\right]\leq \exp\left(-\frac{\varepsilon^2\E[S_n]}{3a}\right),\text{ and}\\
&\Pr\left[S_n\leq (1-\varepsilon)\E\left[S_n\right]\right]\leq \exp\left(-\frac{\varepsilon^2\E[S_n]}{2a}\right).
\end{align*}
\end{lemma}






\section{Data Structure}
\label{sec:datastructure}

The data structure follows the locality-sensitive \emph{mapping}, or \emph{filtering}, framework~\cite{becker2016new,christiani2017set,christiani2017framework}. In this framework, each element $\x$ is mapped to a \emph{set} of filters $F(\x)$.  This is distinct from locality-sensitive hashing, in which $\x$ would be mapped to a single hash value.

While locality-sensitive filtering is distinct from locality-sensitive hashing, preprocessing and searching follows essentially the same high-level idea.  

To search for a query $\q$ we iterate through each filter $f\in F(\q)$.  For each such filter, we test all vectors $\x\in S$ such that $f\in F(\x)$.  In other words, we iterate through each vector that has a filter in common with $\q$, and test the similarity of $\x$ and $\q$.  If we find a sufficiently close $\x$ we return it; otherwise we return failure after exhausting all $f\in F(\q)$.

Thus, the goal of preprocessing is to make it easy to find elements that have a filter in common with a given query.  For each filter $f$ mapped to by some element of $S$, we store a list of all $\x'$ such that $f\in F(\x')$.  These lists can be stored and accessed easily (i.e. in a hash table), and this method takes space linear in $\sum_{x\in S} |F(\x)|$.  We can preprocess quickly by calculating $F(\x)$ for all $\x\in S$.

Thus, the goal of our data structure is to define a randomized mapping of vectors to sets of filters satisfying:
\begin{itemize}[topsep=0pt,noitemsep]
    \item If $\x$ and $\q$ have low similarity, then $F(\x)\cap F(\q)$  is small in expectation (guaranteeing small space and fast execution).
    \item If $\x$ and $\q$ have high similarity, then $F(\x)\cap F(\q)$ is non-empty with high probability (guaranteeing correctness).
\end{itemize}

\paragraph{Computing $F(\x)$: the choice of paths}
Our data structure is based on the Chosen Path data structure of Christiani and Pagh~\cite{christiani2017set}. For our data structure, each $f\in F(\x)$ corresponds to a path, i.e., an ordered sequence $(i_1,\ldots , i_{\ell})\subseteq [d]^{\ell}$ where each $i_j\in [d]$ is the index of one of the $d$ dimensions. The construction of $F(\x)$ ensures that if $f\in F(\x)$, then it must hold that $\x_i = 1$ for all $i\in f$. We say that the path $f$ was \emph{chosen} by $\x$ if $f\in F(\x)$. 

The data structure comes with a (deterministic) function $s$ which maps each vector $\x\in \{0,1\}^d$, path-length $j$ and bit $i\in [d]$ to a \emph{threshold} $\st{\x}{j}{i}\in [0,1]$. 
The choice of $s$ depends on the desired application and will be specified later on.  In particular, $s$ is how our data structure adapts to the distribution---previous data structures essentially used a constant function for $s$.\footnote{As mentioned previously, this is not the only distinguishing technical detail.  We must also sample without replacement and adjust the path length dynamically based on the probabilities of the sampled bits.}

When initializing our data structure, we once and for all select $k$ hash functions, $h_1,\ldots, h_k$, where $h_j:[d]^j\rightarrow [0,1]$, each chosen independently from a family $\mathcal{H}$ of pairwise independent hash functions. These hash functions are fixed throughout.

We now explain how to recursively compute the set of paths $F(\x)$. Initially, let $F_0(\x)$ consist of the empty path.  We recursively grow the paths one step at a time; in particular, $F_j(\x)$ contains paths of length $j$.  
To demonstrate our recursive process, let $v=(i_1,\ldots , i_j)$ be a path in $F_{j}(\x)$.  
If $\prod_{k=1}^j p_{i_k} \leq 1/n$, we stop recursing, and $v$ is a filter of $\x$.
Otherwise, 
we independently consider each set bit $i$ of $\x$ which is not already in $v$.  With probability $s(\x,j,i)$, we concatenate $i$ to the end of $v$; this results in a new filter $v' \in F_{j+1}(\x)$ with $v' = v\concat i$ (where $\concat$ denotes concatenation). This probabilistic choice is made using $h_{j+1}(v\concat i)$.

We formally define this recursive process using the following equation.

\begin{equation*}
F_{j+1}(\x)=\left\{v\concat i \,\middle|\, \begin{aligned}&v=(i_1,\ldots , i_j)\in F_{j}(\x),\text{}\,\\& \prod_{i_k=1}^j p_{i_k} > 1/n,\text{}\\ & i\in \x\setminus v,\text{} \\&h_{j+1}(v\concat i) < \st{\x}{j}{i} \end{aligned}\right\}.
\end{equation*}

Finally, we define $F(\x)$ to be the union of all paths that stopped recursing.
\[
    F(\x)=\bigcup_{j=1}^k \left\{v=(i_1,\ldots , i_j)\in F_j(\x)\,\middle|\,  \prod_{k=1}^j p_{i_k} \leq 1/n\right\}.
\]

We can calculate $F(\x)$ in $O(d|F(\x)|)$ time.

\paragraph{Preprocessing} 
During the preprocessing, we randomly select the $k$ hash functions and compute $F(\x)$ for each $\x$. Then, we use a standard dictionary data structure to construct an inverted index such that for each $f\in \cup_{\x\in S} F(\x)$, we can look up $\{\x\in S\colon f\in F(\x)\}$.

\paragraph{Answering a query}
For a given query $\q$ we compute its chosen paths $F(\q)$ as described above (using the same hash functions as in the preprocessing). For each $f\in F(\q)$, we then compute $B(\x,\q)$ for every $\x$ which chose the path $f$, i.e., for which $f\in F(\x)$. If an $\x$ with similarity at least $b_1$ is found then we return $\x$. If we have exhausted all candidates without finding such an $\x$, we report that no high-similarity vector was found.


\section{Correctness}

We begin with a structural lemma on the number of paths two vectors $\x$ and $\q$ have in common.  We will use this lemma to prove correctness of both of our data structures.  This lemma follows the same high-level idea of the proof of correctness used in~\cite{christiani2017set}, but must be generalized to handle the distribution-dependent choices made by the data structures.  The proof has been moved to Section~\ref{sec:omittedproofs}.

\begin{lemma}
\label{lem:cor}
Suppose that for $\x\in \{0,1\}^d$ and $\q\in \{0,1\}^d$ the following holds: For every $1\leq j\leq k$ and every $v=(i_1,\ldots , i_j)\in [d]^j$, we have
\begin{align}
    \sum_{i\in (\x\cap\q)\setminus v} \min\{\st{\x}{j}{i}, \st{\q}{j}{i}\}\geq 1.
    \label{eq:branch}
\end{align}
Then, $\Pr[F(\x)\cap F(\q)\neq \emptyset]\geq 1/\log n$.
\end{lemma}

The following lemma is used to prove performance of our algorithms.  We use an inductive argument to bound the number of paths generated by the data structure.   This argument depends crucially on both the thresholds $s(\x,j,i)$ and the distribution-dependent stopping rule.

\begin{lemma}
    \label{lem:per}
Let $\x\in \{0,1\}^d$, and let $\rho$ be such that $\sum_{i\in\x}p_i^\rho \stx \leq c$. Then $\E[\ab{F(\x)}]=O(n^{\rho}c^{\log n})$. Furthermore, the expected time spend on computing $F(\x)$ is $O(n^\rho c^{\log n}\ab{\x})$
\end{lemma}

\begin{proof}
For $v\in F_j(\x)$, define the random variables $Y(\x,v,i)=\mathbf{1}_{h_{j}(v\concat i)\leq \stx}$. Let $F^t_j(\x)\subseteq F_j(\x)$ be the set of paths $v\in F_j(\x)$ such that $\sum_{i\in v} \log 1/p_i \leq t$.  Furthermore, let $F^t_j(\x,i)\subseteq F^t_j(\x)$ be the set of paths $v\in F^t_j(x)$ for which $i\notin v$.
We claim that $\E[\ab{F^t_j}]$ is at most $2^{\rho t}c^{j}$. The proof is by induction over $j$ and $t$. Note that for every $v\in F^t_{j+1}$, there must exist $i\in\x$ and $v'\in F^{t-\log(1/p_i)}(\x,i)$ such that $v=v\concat i'$ and $Y(\x,v,i)=1$. Thus,
\begin{align*}
    \E[\ab{F^t_{j+1}}]&=\E\left[\sum_{i\in \x}\sum_{v\in F^{t-\log (1/p_i)}_{j}(\x,i)}Y(\x,v,i)\right] \\
    &=\E\left[\sum_{i\in \x}\sum_{v\in F^{t-\log (1/p_i)}_{j}(\x,i)}\E[Y(\x,v,i)]\right]\\
    &=\E\left[\sum_{i\in \x}\sum_{v\in F^{t-\log (1/p_i)}_{j}(\x,i)}\stx\right] \\
    &\leq\sum_{i\in\x}\E[\ab{F^{t-\log(1/p_i)}_{j}(\x,i)}]\stx\\
&\leq\sum_{i\in\x}\E[\ab{F^{t-\log(1/p_i)}_{j}(\x)}]\stx \\
    &\leq\sum_{i\in \x}2^{\rho (t-\log(1/p_i))}c^{j} \stx\\
&\leq 2^{\rho t}c^j \sum_{i\in \x}p_i^\rho \stx \\
    &\leq 2^{\rho t}c^{j+1}.
\end{align*}
    Now, for every $v\in F(\x)$ there must exist $j$ and $i\in\x$ such that $v=v'\concat i$ for some $v'\in F_j^{\log n}(\x)$. (The $\log n$ follows from taking the log of both sides of $\prod_{i_k} p_{i_k} \leq 1/n$). It follows that
\begin{align*}
\E[\ab{F(\x)}]&=\E\left[\sum_{j=0}^{\log n} \sum_{i\in \x}\sum_{v\in F^{\log n}_{j}(\x,i)}Y(\x,v,i)\right] \\
    & \leq \sum_{j=0}^{\log n} \E[\ab{F_j^{\log n}}] \sum_{i\in x}\stx\\
&\leq 2^{\rho\log n} c^{\log (n)+1}\sum_{i\in\x}\stx \\
    &=O(n^{\rho}c^{\log n}).
\end{align*}
We now bound the expected time for computing $F(\x)$. To this end, note that we spend at most $O(\ab{\x})$ time in each recursive step, and that the expected number of recursive steps is at most $\sum_{j=0}^{\log n} \E[\ab{F^{\log (n)}_j}] =O(n^{\rho}c^{\log n})$.
\end{proof}

The next lemma shows that because we stop each branch of our recursive process once the expected number of vectors from $S$ which choses this path is constant, it follows that the expected query time is linear in $\ab{F(\q)}$.

\begin{lemma}
Let $\q\in\{0,1\}^d$, and let $\rho$ be such that $\sum_{i\in\x}p_i^\rho \st{\q}{j}{i} \leq c$. Furthermore, let $S\sim \dis^n[p_1,\ldots , p_d]$. Then $\E[\sum_{\x\in S} \ab{F(\q)\cap F(\x)}]=O(n^{\rho}c^{\log n})$.
\end{lemma}
\begin{proof}
From Lemma \ref{lem:per}, we know that $\E[\ab{F(\q)}]=O(n^{\rho}c^{\log n})$. Let $v=(i_1,\ldots , i_j)\in F(\q)$ be any path chosen by $\q$. By definition, $\Pr[v\in F(\x)]\leq \Pr[\x_{i_1}=1\land \cdots\land \x_{i_j}=1] \leq 1/n$. Thus,
\begin{align*}
\E\left[\sum_{\x\in S} \ab{F(\q)\cap F(\x)}\right]&=\E\left[\sum_{v\in F(\q)}\sum_{\x\in S}\mathbf{1}_{\{v\in F(\x)\}}  \right] \\
    & =\E\left[\sum_{v\in F(\q)}\sum_{\x\in S}\E[\mathbf{1}_{\{v\in F(\x)\}}]\right]\\
&\leq \E\left[ \sum_{v\in F(\q)}\sum_{\x\in S}1/n\right]\\
    & = \E[\ab{F(\q)}]\\
    & =O(n^{\rho}c^{\log n}).\qedhere
\end{align*}
\end{proof}

\section{Adversarial Queries}
\label{sec:adversarial}

For the adversarial case, we set $\st{\x}{j}{i}=\frac{1}{b_1\ab{\x}-j}$. Thus, the sampling thresholds only depend on $\ab{\x}$ and the number of bits already contained in the path $v$.  The remainder of the data structure follows the description in Section~\ref{sec:datastructure}.

\subsection{Correctness}
We begin with a proof of correctness, giving a guarantee that if there exists an $\x$ similar to a query $\q$, then $\x$ and $\q$ will share a filter (so $\x$ will be found by our algorithm).  

Fix $\x$ and $\q$ such that $B(\x,\q)\geq b_1$ and fix $v=(i_1,\ldots , i_j)\in [d]^j$. The assumption $B(\x,\q)\geq b_1$ is equivalent to $\ab{\x\cap\q}\geq b_1\max\{\ab{\x},\ab{\q}\}$. Thus,
\begin{align*}
    \sum_{i\in (\x\cap\q)\setminus v} \min\{\st{\x}{j}{i}, \st{\x}{j}{i}\}
    &={} \sum_{i\in (\x\cap\q)\setminus v}\frac{1}{b_1\max\{\ab{\x},\ab{\q}\}-j}\\
    &={} \frac{\ab{\x\cap\q}-j}{b_1\max\{\ab{\x},\ab{\q}\}-j}\\
    &\geq{} 1.
\end{align*}

Correctness follows immediately from Lemma~\ref{lem:cor}.

\subsection{Performance guarantees}

\paragraph{Query time}
We bound the query time using a constant $\rho$ that depends only on $b_1$ and $\dis$.

\begin{lemma}
For every $\varepsilon>0$ there exists a constant $C$ such that if $\sum_{i\in\q}p_i^\rho\leq b_1\ab{\q}$ and $\ab{\q}\geq C\log n$, then $\E[\ab{F(\q)}]=O(n^{\rho+\varepsilon})$.
\end{lemma}
\begin{proof}
First, if $\ab{\q}\geq C\log n$ for some $C>1$, then
\begin{align*}
\st{\q}{j}{i}&=\frac{1}{b_1\ab{\q}-j}\\
    &\leq\frac{1}{b_1\ab{\q}-\log n}\\
    &= \frac{1}{b_1\ab{\q}}\frac{1}{1-(\log n)/(b_1\ab{\q})}\\
    &\leq \frac{1}{b_1\ab{\q}}\frac{1}{1-1/(b_1C)}.
\end{align*}
Thus,
\begin{align*}
    \sum_{i\in \q}p_i^\rho \st{\q}{j}{i} &\leq \sum_{i\in\q}\frac{p_i^\rho}{b_1\ab{\q}}\frac{1}{1-1/(b_1C)}\\
    &\leq \frac{1}{1-1/(b_1C)}. 
\end{align*}
It follows from Lemma~\ref{lem:per} that $\E[\ab{F(\q)}]=O\left(n^{\rho}\left(\frac{1}{1-1/(b_1C)}\right)^{\log n}\right)$.
\end{proof}

\paragraph{Preprocessing time} 
The following lemma bounds our expected per-element processing time.  The proof has been moved to Section~\ref{sec:omittedproofs}.

\begin{lemma}
    \label{lem:preprocess}
For every $\varepsilon>0$ there exists a constant $C$ such that if $\sum_{i\in[d]}p_i^{1+\rho}= b_1\sum_{i\in [d]}p_i$ and $\sum_{i\in[d]}p_i\geq C\log n$, then $\E[\ab{F(\x)}]=O(n^{\rho+\varepsilon})$ for $\x\sim\dis[p_1,\ldots , p_d]$.
\end{lemma}

We multiply this by $n$ to get the expected total preprocessing time for all elements.  Since our data structure takes space linear in the total size of the stored filters, we similarly obtain $O(n^{1 + \rho+\varepsilon})$ space.

\section{Correlated Queries }

We need samples from the distribution to have sufficient length such that we can guarantee that $B(\q,\x) > B(\q,\x')$ for $\x'\sim\dis$ which is not correlated with $\q$ (see Lemma~\ref{lem:correlatedclose}).  
Therefore, we assume there is a sufficiently large constant $C$ that the expected size of an element drawn from $\dis$ is at least $\sum_{i\in [d]} p_i \geq C\log n$.    Similarly, we assume that all $p_i\leq \alpha/2$.
We assume that $C \gg 1/\alpha$.

In the correlated case, we no longer need to sample vertices uniformly at random.  
The query $\q$ and the distribution $\dis$ give us further information about where $\x$ and $\q$ are likely to intersect. In particular, we can calculate the conditional probability
\[
\hat{p}_i = \Pr(\x_i = 1 ~|~ \q_i = 1) = p_i(1-\alpha) + \alpha. 
\]
We weight our chosen path choices by this conditional probability.  We then increase each sampling probability by a small constant $1 + \delta = 1 + 3/(\sqrt{\alpha C})$; this will help us ensure correctness for filters.  (We derive this constant in the proof of Lemma~\ref{lem:correlatedcorrectness}.  A smaller constant is likely sufficient in practice, particularly for small $\alpha$.)
Thus, at round $j$, we sample each bit $i \in \x$, without replacement, with probability 
\[
    \st{\x}{j}{i} = \frac{1 + 3/\sqrt{\alpha C} }{\hat{p}_i C(\log n) - j}.
\] 

With these new sampling probabilities, we again maintain paths as defined in Section~\ref{sec:datastructure}.  To answer a query $\q$, we look through all $\x'\in F(\q)$, returning any $\x'$ that has similarity at least $b_1 = \alpha/1.3$.

\subsection{Correctness}
\label{sec:correlatedanalysis}

For a query $\q$, our data structure attempts to find the $\x$ that is $\alpha$-correlated with $\q$ by finding a similar vector among the items stored for filters in $F(\q)$.  We begin by showing that with high probability, $B(\x,\q) > b_1$; meanwhile, for uncorrelated $\x'$, $B(\x',\q) < b_1$. 

\begin{lemma}
    \label{lem:correlatedclose}
    Assume $q\sim\dis_\alpha(\x)$.
    With high probability, $B(\x,\q) \geq \alpha/1.3$.  Meanwhile, for all $\x'\in S$ not correlated with $\q$, $B(\x',\q) \leq \alpha/1.5$ with high probability.
\end{lemma}

\begin{proof} 
    To begin, note that for any $\x'\sim \dis$, $\E[\ab{\x}] = C\log n$.  By Chernoff bounds (Lemma~\ref{wchernoff} for example) and the union bound, $\sqrt{1.3}(C\log n)\geq \min\{\ab{\x'},\ab{\q}\} \geq (C\log n)\sqrt{3}/2$ with high probability for all $\x'$ and $\q$.  

    First, consider the correlated pair: $\q\sim\dis_\alpha(\x)$.  We have $\E[\ab{\x\cap\q}] = \sum_i p_i^2(1 - \alpha) + p_i\alpha \geq \alpha C\log n$.  Again by Chernoff, $\ab{\x\cap\q} \geq \alpha (C\log n)/\sqrt{1.3}$ with high probability.  Combining this with the above, $B(\x,\q)\geq \alpha/1.3$ with high probability.  

    Now, consider an uncorrelated pair $\x'$ and $\q$ drawn independently from $\dis$.  We have $\E[\ab{x\cap\q}] = \sum_i p_i^2 \leq (\alpha C\log n)/2$.  Using Chernoff bounds, $\E[\ab{x\cap\q}] \leq (2/\sqrt{3})  (\alpha C\log n)/2$ with high probability.  Combining with the above, $B(\x,\q)\leq \alpha/1.5$ with high probability.  
\end{proof}

Now we show that our sampling probabilities are large enough to guarantee that the two $\alpha$-correlated vectors $\x$ and $\q$ are likely to share a path.   We need to use new techniques beyond those in Section~\ref{sec:adversarial} (and beyond previous results) because we our proof must leverage that the close pair is chosen randomly.  After all, if we ignore this aspect, we have reduced to the adversarial case---we want to improve those bounds.

We do this by showing that, with high probability, any given path during our recursive process satisfies the requirements of Lemma~\ref{lem:cor}.  Note that we prove correctness with probability at least $1 - 1/n^2$ for simplicity; we can obtain stronger bounds by slightly increasing $\delta$.

\begin{lemma}
    \label{lem:correlatedcorrectness}
    Consider a path $v$ of length at most $\log n$, and $q\sim \dis_\alpha(\x)$.  
    Assume $C\alpha \geq 15$.  
    Then, with probability at least $1 - 1/n^2$, 
    \[
        \sum_{i\in (\x\cap \q)\setminus v} s({\x},\ab{v},i) \geq 1.
    \] 
\end{lemma}

\begin{proof} 
    We begin by calculating the expected value.
    \begin{align*}
        \E\left[ \sum_{i\in \filterint} s({\x},\ab{v},i)\right] &= \sum_{i\in[d]\setminus v} \Pr(i\in\x\cap\q) \frac{1+\delta}{\hat{p_i}C(\log n) - j}\\
        & \geq  \sum_{i\in[d]\setminus v}  p_i\frac{(1+\delta)}{C(\log n) - j/\hat{p_i}} \\
        & \geq \frac{(1+\delta)(C\log n - \sum_{i\in v} p_i)}{C(\log n) - j}\\
        & \geq 1 +\delta.
\end{align*}
We have $\hat{p_i}\geq \alpha$, so 
$s(\x,\ab{v},i) \leq (1 + \delta)/(\alpha C\log n)$.  Then by Lemma~\ref{wchernoff},
    \[
    \Pr\left(\sum_{x\in \filterint} s(\x,\ab{v},i) \leq 1\right) \leq \exp\left(-\frac{\alpha C\log n}{2} \left(\frac{\delta}{1+\delta}\right)^2 \right).
\]
Assume that $C\alpha \geq 15$.  Then we assign 
\[
    \delta = \frac{3}{\sqrt{\alpha C}} \geq \frac{2\sqrt{\ln 2/(\alpha C)} }{1 - 2\sqrt{\ln 2/(\alpha C)}}.
\]
Substituting, we obtain the lemma.  
\end{proof}

Applying union bound, all paths in our recursive process satisfy Lemma~\ref{lem:cor}.  
Thus, these two lemmas imply that we return the $\alpha$-correlated $\x$ with high probability.

\subsection{Performance Guarantees}
We now compute the expected number of paths chosen by an $\x\sim \dis$.  Since we have both $\q\sim\dis$ and $\x\sim\dis$ for all $\x\in S$, this lemma implies the query time and space bounds of Theorem~\ref{thm:alpha-correlated} immediately.

\begin{lemma}
    \label{lem:correlatedspeed}
For every $\varepsilon>0$ there exists a constant $C$ such that if $\sum_{i\in[d]}\dfrac{p_i^{1+\rho}}{\hat{p}_i}=\sum_{i\in [d]}p_i$ and $\sum_{i\in[d]}p_i = C\log n$, then $\E[\ab{F(\x)}]=O(n^{\rho+\varepsilon})$ for $\x\sim\dis[p_1,\ldots , p_d]$.
\end{lemma}

\begin{proof}
We want to show that if $C$ is sufficiently large, then $\E_{\mathcal{H},\x\sim\dis}[\ab{F(\x)}]=O(n^{\rho+\varepsilon})$ where the expectation is over both the randomness $\mathcal{H}$ of the data structure and the randomness of $\x$.
Since $\hat{p}_i\geq \alpha$ and $j\leq \log n$, we have
\begingroup
\interdisplaylinepenalty=10000
\begin{align*}
\st{\x}{j}{i}&=\frac{1}{\hat{p}_i C\log n-j}\\
    &\leq\frac{1}{\hat{p}_i C\log n-\log n}\\
    &\leq \frac{1}{\hat{p}_i C\log n}\cdot \frac{1}{1-\frac{\log n}{\hat{p}_i C\log n}}\\
&\leq \frac{1}{\hat{p}_i C\log n}\frac{1}{\alpha C}.
\end{align*}
Furthermore, by Lemma~\ref{wchernoff}, 
\begin{align*}
\Pr\left[\sum_{i\in\x}p_i^\rho / \hat{p}_i\leq (1+C^{-1/3})\sum_{i\in[d]} p_i^{1+\rho} / \hat{p}_i\right]
&\geq 1-\exp\left(-\alpha\frac{(C^{-1/3})^2}{3}\sum_{i\in [d]}p_i^{1+\rho} / \hat{p}_i\right)\\
& \geq 1-\exp\left(-\alpha\frac{C^{-2/3}}{3}C\log (n)\right)\\
& \geq 1-n^{\alpha\frac{b_1C^{1/3}}{3}}.
\end{align*}
\endgroup
Thus, with probability $1-n^{\Omega(C^{1/3})}$, we have 
\begin{align*}
\sum_{i\in\x}p_i^\rho \st{\x}{j}{i}&\leq \sum_{i\in\x}p_i^\rho \frac{1}{\hat{p}_i C\log n}\frac{1}{\alpha C}\\
&\leq \frac{1+C^{-1/3}}{C\log (n)\alpha C}\sum_{i\in [d]} p_{i}^{1+\rho} / \hat{p}_i\\
&=\frac{1+C^{-1/3}}{\alpha C}.
\end{align*}
Lemma~\ref{lem:per} yields $\Pr_{\x\sim \dis}[\E_{\mathcal{H}}[F(\x)]\leq n^{\rho+\varepsilon}]\geq 1-n^{\Omega(C^{1/3})}$. Since $\E_\mathcal{H}[\ab{F(\x)}]$ is polynomial in $n$ for any $\x$, this implies that $\E_{\mathcal{H},\x\sim \dis}[\ab{F(\x)}]=O(n^{\rho+\varepsilon})$.
\end{proof}

\section{Performance Comparison}
\label{sec:examples}

The bounds achieved by our data structures are given as the solution to an equation which is not in closed form.  In this section, we give some examples and intuition of how much speedup we achieve.

\subsection{Adversarial query}
If the query is adversarial, then the query time for a query $\q$ is determined by the smallest $\rho$ such that 
\begin{equation*}
\sum_{i\in\q}p_i^\rho\leq b_1\ab{\q}.
\label{ourrho}
\end{equation*}
We will now discuss to what extent our data structure is able to take advantage of possible skew in the item-level probabilities. In order to simplify the discussion, assume that $\ab{\q}=\sum_{i\in[d]}p_i$ (so that $\ab{\q}$ equals the expected value of $\ab{\x}$ when $\x\sim\dis$, meaning that all sets have roughly the same size). Now, for every $\varepsilon>0$ there exists a $C$ such that if $\sum_{i\in [d]}p_i\geq C\log n$, then $B(\x,\q)=\frac{1}{\ab{\q}}\sum_{i\in\q}p_i\pm\varepsilon$ for every $\x\in S$ with high probability. Thus, we can solve the problem by solving the $(b_1,b_2)$-approximate Braun-Blanquet similarity search problem using the standard Chosen Path data structure, thereby obtaining a $\rho$-value which, for sufficiently large $C$, is arbitrarily close to
\begin{equation*}
    \rho_{\texttt{CP}}=\frac{\log(b_1)}{\log \left(\frac{1}{\ab{\q}} \sum_{i\in\q}p_i \right)}.
\end{equation*}
If the part of the input distribution relevant for the query $\q$ contains no skew, i.e., if $p_i=p$ for all $i\in\q$, then $\rho_{\texttt{CP}}=\rho$. However, $\rho<\rho_{\texttt{CP}}=\log(b_1) / \log (p)$ in all other cases. 

To our knowledge there is no closed-form expression for the solution $\rho$ to the equation $\sum_{i\in\q}p_i^\rho\leq b_1\ab{\q}$. Instead, we will compute $\rho$ in a few settings to illustrate how skew influences our query time. Suppose that $\q$ consists of two types of bits: half the bits of $\q$ are set with probability $p_a$ in a random set from $S$, and the other half is set with probability $p_b$. Furthermore, assume that $\sum_{i\in[d]}p_i=\ab{\q}=\Theta(\log n)$.

Suppose first that the distribution has very significant skew, e.g., $p_a=1/4$, $p_b= n^{-0.9}$. Assume that we are searching for a set in $S$ with similarity at least, say, $b_1=1/3$. In this case, $\rho_\texttt{CP}\geq \frac{\log (1/3)}{\log(1/8)}\geq 0.528$, whereas we obtain a $\rho$-value of $\rho=\frac{\log(2/3)}{\log(1/4)}+o_n(1)\leq 0.293$ for $n$ sufficiently large. Prefix filtering does not give any non-trivial (worst-case) performance guarantee. This example shows that we can take advantage of the skew even when it is possible that the entire intersection between $\q$ and the close point $\x$ is within a part of $\q$ where all bits are set with the same probability in a datapoint. A more extreme example occurs if we increase $b_1$ to $b_1=2/3$ so that a possible close point $\x$ must share some of the bits for which the associated probability is $p_b=n^{-0.9}$. In this case, equation (\ref{ourrho}) is satisfied for $\rho$ arbitrarily close to zero. Thus, we achieve a very fast query time (in particular, our query time will be $O(n^{\varepsilon})$ for every constant $\varepsilon>0$). In order to understand why, note that no path in $F(\q)$ can contain more than two bits with associated item-level probability $n^{-0.9}$. Furthermore, since $b_1>1/2$, our sampling threshold $s(\q,j,i)$ is low enough that we are very unlikely to create long paths consisting only of bits with associated probability $1/4$. Indeed, the expected size of a path in $F(\q)$ is $O(1)$ in this case. Solving the problem by solving the $(b_1,b_2)$-approximate problem will yield a $\rho$ value of $\rho_{\texttt{CP}}=\log(2/3)/\log (1/8)=0.194$. Prefix filtering will need $\Omega(n^{0.1})$ time to locate a possible close point if it exists.

\subsection{Correlated query}
When the query is $\alpha$-correlated with a point $\x$ in our dataset $S$, the running time is (for $\sum_{i\in[d]}p_i$ sufficiently large) determined by the smallest $\rho$ such that $\sum_{i\in[d]} p_i^{1+\rho}/{\hat{p}_i}\leq \sum_{i\in [d]}p_i$ where $\hat{p}_i=(1-\alpha)p_i+\alpha$.

As for the adversarial case, whenever there is significant skew in the input distribution, we can get very fast query time. Suppose for example that $\dis[p_1,\ldots ,p_d]$ is such that $4 C\log (n)$ bits are set to $1$ with probability $p_a=1/4$ and $n^{9/10}C\log (n)$ bits are set to $1$ with probability $p_b=n^{-9/10}$. If $\alpha=2/3$, then we find that our expected query time is $O(n^\varepsilon)$ for every constant $\varepsilon>0$. On the other hand, prefix filtering takes $\Omega(n^{0.1})$ time.

Let us now consider examples where the input distribution is skewed but all probabilities are $\Omega(1)$, meaning in particular that prefix filtering does not give any non-trivial worst-case guarantees. For example, suppose that $C\log n$ bits are set to $1$ with probability $p_a=\Theta(1)$ and $C\log n$ bits are set to $1$ with probability $p_b=\Theta(1)$. We compare the performance of our data structure to that of using Chosen Path for solving the $(b_1,b_2)$-approximate problem where $b_1$ is the expected similarity between the correlated points and $b_2$ is the expected similarity between the query and an uncorrelated point. For $p_a=p$ and $p_b=p/8$ and $\alpha=2/3$, we plot the corresponding $\rho$-values in Figure~\ref{figure}.


\section{Skew in Real Data Sets}
\label{sec:experiments}

\begin{figure*}[t]
\centering
\includegraphics[width=.49\textwidth]{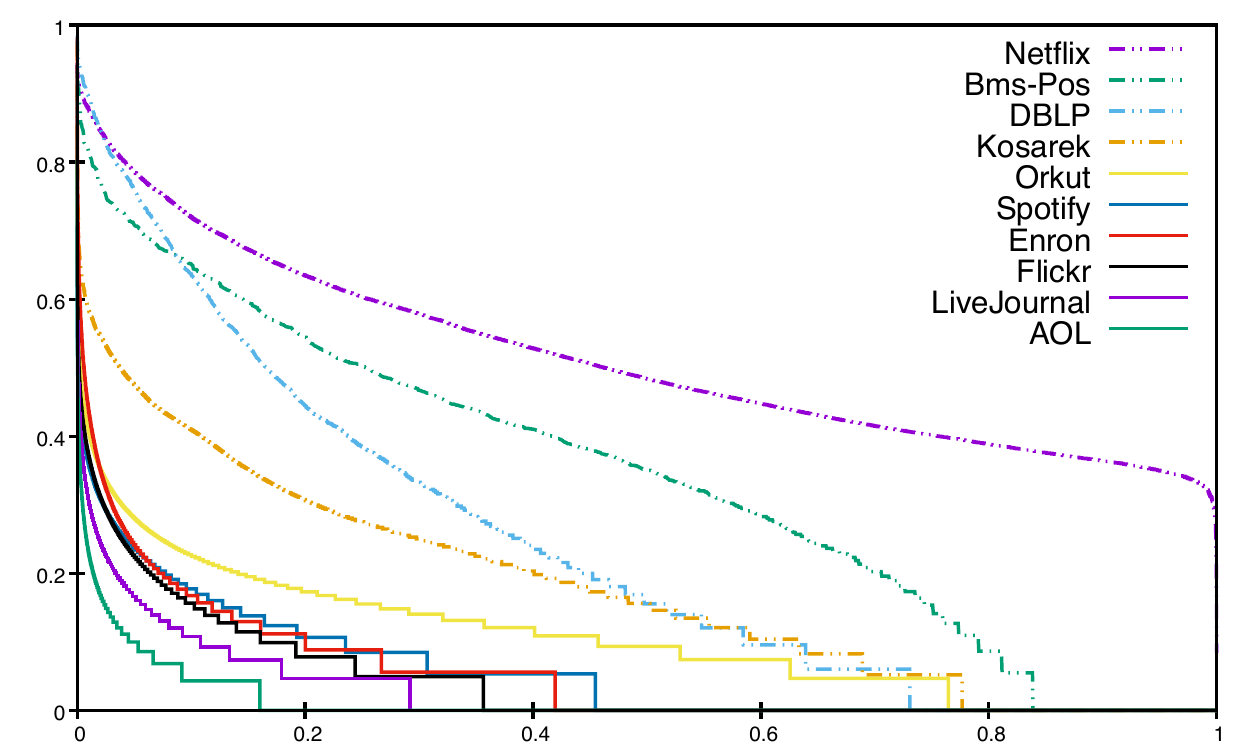}
\includegraphics[width=.49\textwidth]{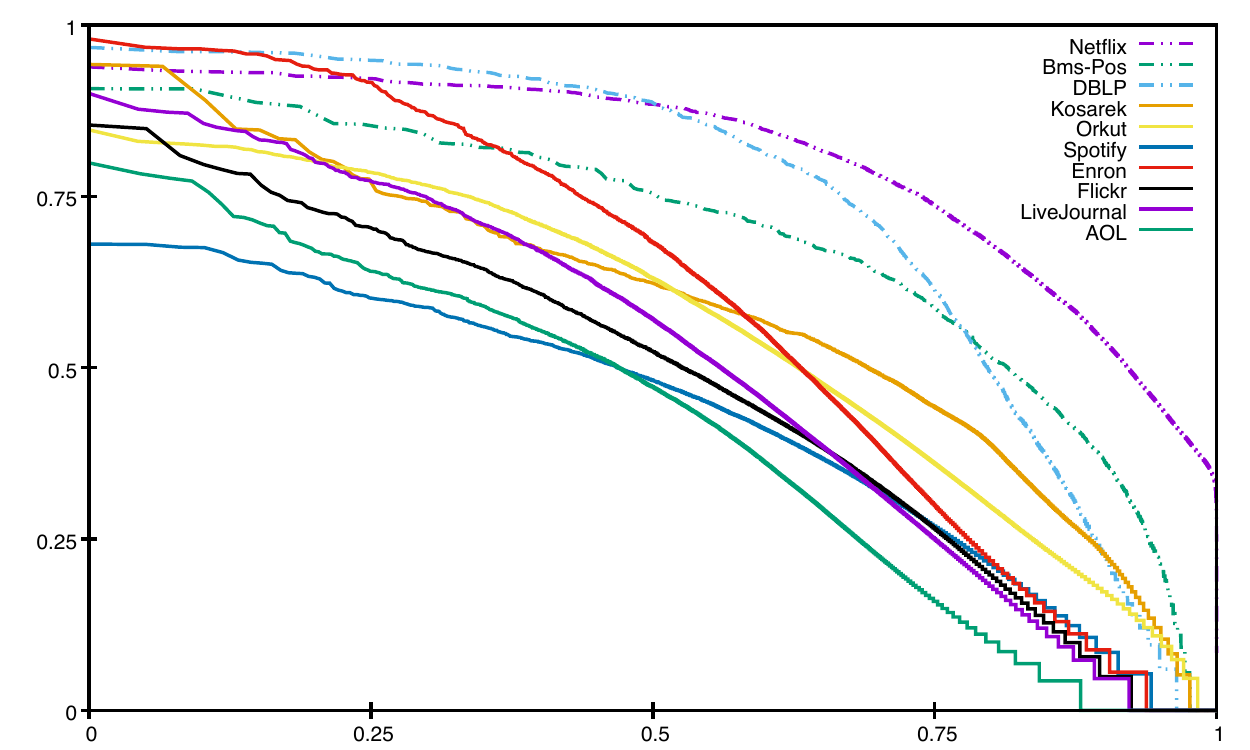}
\caption{Frequency distributions of real datasets from the set similarity search benchmark of Mann et al.~\cite{Mann2016}, with frequencies $p_j$ in decreasing order, plotted in two ways.
In both figures, for each $j\in [d]$, the y-axis denotes $1+\log_n p_j$.  
On the left, the x-axis denotes $j/d$; on the right, the x-axis denotes $\log_d j$.}
\label{fig:zipf}
\end{figure*}

In this section we examine how the ideas that inspired our model apply to the real-life sparse data sets used to evaluate similarity set methods by Mann et al.~\cite{Mann2016}.

\medskip

\paragraph{Skew}
For each data set we computed the empirical frequencies $p_j$, indexed such that frequencies decrease as $j$ grows.
In Figure~\ref{fig:zipf} we show the distribution of frequencies $p_j$ in two ways.
On the left side we plot $\log_n(p_j n)$ against $j/d$, where $n$ is the number of vectors and $d$ is the dimension of the vectors.
On the right side we show the distribution on a normalized log-log scale.
As can be seen, all data sets display a significant skew. 
A plain Zipfian distribution would appear linear on this plot (with slope corresponding to the exponent in the distribution).
Frequencies do not follow a Zipfian distribution, but are generally close to a ``piecewise Zipfian'' distribution.
For almost all data sets, the frequencies outside of the very top can be thought of as being bounded by $p_j \leq n^{-\gamma}$ for some constant $\gamma > 0$.


\paragraph{Approximate Independence}
Our theoretical analysis is based on the assumption of independence between the set bits (i.e. the items) of the random vectors, and more specifically on the inequality:
\begin{equation}\label{eq:indep}
	\Pr_{\x\in S}[\forall_{j\in I} \x_j=1] \leq \prod_{j\in I} p_j, \text{ for } I\subseteq [d] \enspace .
\end{equation}
The asymptotic analysis still applies if (\ref{eq:indep}) holds only up to a constant factor.
(Under some conditions we may even let this factor increase exponentially with $|I|$.)
To shed light on whether this assumption is realistic we considered random size-2 and size-3 subsets $I$, respectively, and computed the expected number of vectors $\x$ satisfying $\forall_{j\in I} \x_j=1$ under the independence assumption and in the data set, respectively.
The ratio of these numbers is the constant factor needed, on average over all sets $I$, to make (\ref{eq:indep}) hold. 

\begin{wraptable}{r}{.4\textwidth}
	\centering
	\begin{tabular}{lrr} \hline
		Data set & $|I|=2$ & $|I|=3$ \\
		\hline
		AOL        & $1.2$ &  $3.9$ \\
		BMS-POS    & $1.5$ &  $3.9$ \\
		DBLP       & $1.4$ &  $2.3$ \\
		ENRON      & $2.9$ & $21.8$ \\
		FLICKR     & $1.7$ &  $4.9$ \\
		KOSARAK    & $7.1$ &$269.4$ \\
		LIVEJOURNAL& $2.3$ &  $7.3$ \\
		NETFLIX    & $3.1$ & $24.0$ \\
		ORKUT      & $4.0$ & $37.9$ \\
		SPOTIFY    & $24.7$& $6022.1$ \\
		\hline
	\end{tabular}
	\caption{We computed the ratio bewteen $\E_I[\Pr_{\x\in S}[\forall_{j\in I} \x_j=1]]$ (expected observed number of vectors with 1s in $I$) and $\E_I[\prod_{j\in I} p_j]$ (expected number of vectors with 1s in $I$ assuming independence). The expectations are computed for $I$ chosen uniformly from subsets of $[d]$ of size 2 and 3, respectively.}
	\label{tab:independence}
        \vspace{-.5in}
  \end{wraptable}

The ratios are shown in Table~\ref{tab:independence}. If the independence assumption were true, the ratios would be close to 1. As can be seen all data sets have some kind of positive correlation between dimensions meaning that there are more pairs/triples than the independence model predicts. 
However, for most of the data sets it seems reasonable to assume that these correlations are weak enough that (\ref{eq:indep}) holds up to a constant factor independent of $n$, at least for small constant~$|I|$.
In those cases we expect our theoretical analysis based on independence to be indicative of actual running time.

\section{Conclusion}
Several open problems remain about how to handle nearest neighbor search in skewed datasets.

One natural question is how to relax the assumption that all item probabilities $p_i$ are known to the algorithm beforehand.  It seems likely that one can estimate each $p_i$ to very high precision by counting the occurrences in the dataset itself, leading to the same asymptotic bounds.  We note that the set similarity join algorithm in~\cite{ChristianiPaSi17} avoids using or estimating $p_i$, but we do not know how to analyze its performance in our setting.

The performance examples in this paper deal with distributions with two types of item: very rare and very common items.  In practice, one more often encounters distributions with much more gradual skew, such as a Zipf distribution.  Unfortunately, sets selected using a Zipf distribution have very small expected size, which trivializes the asymptotics.  It would be interesting to find a class of distributions that accurately characterizes the skew of real data while remaining interesting for asymptotic analysis.  Such a distribution would be an important use case for our algorithm.

Finally, it would be interesting to relax the independence assumption.  Our experiments gave some evidence that many of the datasets we studied have only mild dependencies, but for some this was not the case (particularly the Spotify dataset).  In fact, the Spotify dataset has recently been observed to be a difficult case for a variant of the Chosen Path algorithm~\cite{ChristianiPaSi17}, possibly due to these correlations.  It seems that if the correlations are ``simple'' and known ahead of time, there may be strategies to deal with them when sampling paths.  Such an algorithm would loosen the independence assumption in our model, and has the potential to lead to much stronger performance in practice.

\section{Omitted Proofs}
\label{sec:omittedproofs}

\begin{proof}[Proof of Lemma~\ref{lem:cor}]
In order to simplify calculations, we are going to consider a conveniently chosen subset of $F(\x)\cap F(\q)$. In order to define this subset, for $v\in F_j(\x)\cap F_j(\q)$ and $i\in (\x\cap\q)\setminus v$ let, $$s(v\circ i)=\frac{\min\{\st{\x}{j}{i}, \st{\q}{j}{i}\}}{\sum_{i'\in (\x\cap\q)\setminus v} \min\{\st{\x}{j}{i'}, \st{\q}{j}{i'}\}}.$$
We define $M_j$ and $M$ exactly as $F_j(\x)\cap F_j(\q)$ and $F(\x)\cap F(\q)$, except that we replace the requirements $h_j(v\concat i)\leq \st{\x}{j}{i}$ and $h_j(v\concat i)\leq  \st{\q}{j}{i}$ with the requirement $h_j(v\concat i)\leq s(v\concat i)$. It follows from the assumption (\ref{eq:branch}) that $M$ is indeed a subset of $F(\x)\cap F(\q)$. We claim that $\Pr[M\neq \emptyset]\geq 1/\log n$. In order to show this, we will make use of the following second moment bound:
\begin{equation}
\Pr [M_j\neq\emptyset]\geq\frac{\E[\ab{M_j}]^2}{\E[\ab{M_j}^2]}.
\label{eq:secondmoment}
\end{equation}
From (\ref{eq:secondmoment}), it suffices to show $\E[\ab{ M_j}^2]\leq j+1$ and $\E[\ab{M_j}]=1$ for all $j\geq 0$. We will prove that this holds by induction. Recall that we start with a single empty path, and so $\ab{M_0}=1$ with probability $1$. Thus, the statement trivially holds for $j=0$. Let $j>0$. 
Define the random variable $Y(v\concat i)=\mathbf{1}_{h_{j}(v\concat i)\leq s(v\concat i)}$. Note that this random variable is independent of $M_{j-1}$. In particular, for every possible $m_{j-1}$, we have $\E[Y(v \concat i)\vert M_{j-1}=m_{j-1}]=\E[Y(v\concat i)]=s(v\concat i)$. Using this independence along with (\ref{eq:branch}) yields 
\begin{align*}
\E[M_j]&=\E\left[\sum_{v\in M_{j-1}}\sum_{i\in (\x\cap\q)\setminus v}Y(v\concat i)\right]\\
&=\E\left[\sum_{v\in M_{j-1}}\sum_{i\in (\x\cap\q)\setminus v}\E[Y(v\concat i)]\right]\\
&=\E\left[\sum_{v\in M_{j-1}}\sum_{i\in (\x\cap\q)\setminus v}s(v\concat i)\right]\\
&\geq \E\left[\sum_{v\in M_{j-1}} 1\right]=\E[\ab{M_{j-1}}].
\end{align*}

    Recall that the hash-function $h_j$ was chosen from a pairwise independent family. Thus, for $v\concat i\neq v'\concat i'$, we have $\E[Y(v\concat i) Y(v'\concat i')] =\E[Y(v\concat i)]\E[Y(v'\concat i')]$. Then,
    \begin{align*}
        \E[\ab{ M_{j}}^2] ={}& \E\left[\left(\sum_{v\in M_{j-1}}  \sum_{i\in (\x\cap\q)\setminus v}Y(v\concat i)\right)^2\right]\\
        ={}& \E\left[\sum_{v\in M_{j-1}, i\in (\x\cap\q)\setminus v}Y(v\concat i)^2\right] 
        + \E\left[\sum_{\substack{v\in M_{j-1}, i\in (\x\cap\q)\setminus v\\u\in M_{j-1}, i'\in (\x\cap\q)\setminus v' \\v\concat i\neq v'\concat i'}} Y(v\concat i) Y(v'\concat i')\right]\\
        ={}& \E\left[\sum_{v\in M_{j-1}, i\in (\x\cap\q)\setminus v}\E[Y(v\concat i)^2]\right] 
        + \E\left[\sum_{\substack{v,v'\in M_{j-1}\\ i,i'\in (\x\cap\q)\setminus v\\v\concat i\neq v'\concat i'}} \E[Y(v\concat i)]\E[Y(v'\concat i')]\right].\\
\end{align*}

    We have $\E[Y(v\concat i)] = s(v\concat i)$.
    Furthermore, $\sum_{i\in (\x\cap\q)\setminus v} s(v\concat i) \leq 1$ (the same applies with $v'$, $i'$ substituted for $v$, $i$ respectively).  Substituting, 
\begin{align*}
        \E[\ab{ M_{j}}^2] ={}& \E\left[\sum_{v\in M_{j-1}, i\in (\x\cap\q)\setminus v}s(v\concat i)\right] 
        + \E\left[\sum_{\substack{v\in M_{j-1}, i\in (\x\cap\q)\setminus v\\u\in M_{j-1}, i'\in (\x\cap\q)\setminus v' \\v\concat i\neq v'\concat i'}} s(v\concat i) s(v'\concat i')\right]\\
        \leq{}& \E\left[\sum_{v\in M_{j-1}} 1\right] + \E\left[ \sum_{v\in M_{j-1}} \sum_{v'\in M_{j-1}}1\right]\\
        ={}& \E[\ab{M_{j-1}}]+\E[\ab{M_{j-1}}^2]\\
        \leq{}& 1+j.\qedhere
    \end{align*}
\end{proof}

\bigskip

\begin{proof}[Proof of Lemma~\ref{lem:preprocess}] 
Let $\x\sim\dis[p_1,\ldots , p_d]$ and assume that $\sum_{i\in[d]}p_i\geq C\log n$ for some $C>1$. We want to show that if $C$ is sufficiently large, then $\E_{\mathcal{H},\x\sim\dis}[\ab{F(\x)}]=O(n^{\rho+\varepsilon})$ where the expectation is over both the randomness $\mathcal{H}$ of the data structure and the randomness of $\x$. By a Chernoff bound, 
\begin{align*}
    \Pr\left[\ab{\x}\geq (1-C^{-1/3})\sum_{i\in[d]}p_i\right]
&\geq 1-\exp\left(-\frac{(C^{-1/3})^2}{2}\sum_{i\in[d]}p_i\right)&\\
&\geq 1-\exp\left(-\frac{C^{-2/3}}{2}C\log (n)\right)\\
&\geq 1-\exp(-C^{1/3}\log(n)/2)\\
&\geq 1-n^{\frac{C^{1/3}}{2}}.
\end{align*}
Furthermore,
\begin{align*}
    \Pr\left[\sum_{i\in\x}p_i^\rho\leq (1+C^{-1/3})\sum_{i\in[d]} p_i^{1+\rho}\right]&
    \geq 1-\exp\left(-\frac{(C^{-1/3})^2}{3}\sum_{i\in [d]}p_i^{1+\rho}\right)\\
    & = 1-\exp\left(-\frac{C^{-2/3}}{3}b_1\sum_{i\in [d]}p_i^{\rho}\right)\\
    & \geq 1-\exp\left(-\frac{C^{-2/3}}{3}b_1C\log (n)\right)\\
    & \geq 1-n^{\frac{b_1C^{1/3}}{3}}.
\end{align*}
 When both of these events occur, we have
\begin{align*}
\st{\x}{j}{i}&=\frac{1}{b_1\ab{\x}-j}\\
    &\leq\frac{1}{b_1\ab{\x}-\log n}\\
    &\leq \frac{1}{b_1(1-C^{-1/3})\sum_{i\in [d]}p_i-\log n}\\
&= \frac{1}{b_1\sum_{i\in [d]}p_i}\frac{1}{1-C^{-1/3} -\log n/(b_1\sum_{i\in [d]}p_i)}\\
&\leq\frac{1}{b_1\sum_{i\in [d]}p_i}\frac{1}{1-C^{-1/3}-1/(b_1C)}.
\end{align*}
And therefore
\begin{align*}
\sum_{i\in\x}p_i^\rho \st{\x}{j}{i}&\leq \sum_{i\in\x}p_i^\rho \frac{1}{b_1\sum_{i\in [d]}p_i}\frac{1}{1-C^{-1/3}-1/(b_1C)}\\
&\leq \sum_{i\in[d]}p_i^{1+\rho} \frac{1}{b_1\sum_{i\in [d]}p_i}\frac{1+C^{-1/3}}{1-C^{-1/3}-1/(b_1C)}\\
&\leq \frac{1+C^{-1/3}}{1-C^{-1/3}-1/(b_1C)}.
\end{align*}
Thus, according to Lemma~\ref{lem:per}, $\Pr_{\x\sim \dis}[\E_{\mathcal{H}}[F(\x)]\leq n^{\rho+\varepsilon}]\geq 1-n^{\Omega(C^{1/3})}$. Since $\E_\mathcal{H}[\ab{F(\x)}]$ is polynomial in $n$ for any $\x$, this implies that $\E_{\mathcal{H},\x\sim \dis}[\ab{F(\x)}]=O(n^{\rho+\varepsilon})$.
\end{proof}

\bibliographystyle{abbrv}
\bibliography{skewnn}

\end{document}